\documentclass[prl,showpacs,amsmath,amssymb,amsfonts,lengthcheck,longbibliography]{revtex4-2}

\usepackage[letterpaper,top=2cm,bottom=2cm,left=3cm,right=3cm,marginparwidth=1.75cm]{geometry}

\usepackage[dvipsnames]{xcolor}
\usepackage{tcolorbox}
\tcbuselibrary{breakable}
\usepackage{graphicx, color, graphpap}% Include figure files
\usepackage{enumitem}
\usepackage{amssymb}
\usepackage{amsthm}
\usepackage{dsfont}
\usepackage{mathtools}
\usepackage{multirow}
\usepackage[colorlinks=true,citecolor=magenta,linkcolor=blue]{hyperref}
\usepackage[T1]{fontenc}
\usepackage{bbm}
\usepackage{thmtools,thm-restate}
\usepackage{verbatim}
\usepackage{mathtools}
\usepackage{titlesec}
\usepackage{amsmath}
\usepackage[normalem]{ulem}
\usepackage{circuitikz}
\usepackage{braket}
\usepackage[caption=false]{subfig} % old package, but only one compatible with RevTex4.2

\usepackage{listings}% http://ctan.org/pkg/listings
\usepackage{csquotes}

\newtheorem{proposition}{Proposition}

\long\def\ca#1\cb{} %Use for commenting out: \ca...\cb

\newcommand{\floor}[1]{\lfloor #1 \rfloor}

\renewcommand{\geq}{\geqslant}
\renewcommand{\leq}{\leqslant}

\begin{document}

\title{Error Mitigation for Thermodynamic Computing}
\author{Maxwell Aifer, Denis Melanson, Kaelan Donatella, Gavin Crooks, Thomas Ahle, and Patrick J. Coles}

\affiliation{Normal Computing Corporation, New York, New York, USA}

\begin{abstract}
    While physics-based computing can offer speed and energy efficiency compared to digital computing, it also is subject to errors that must be mitigated. For example, many error mitigation methods have been proposed for quantum computing. However this error mitigation framework has yet to be applied to other physics-based computing paradigms. In this work, we consider thermodynamic computing, which has recently captured attention due to its relevance to artificial intelligence (AI) applications, such as probabilistic AI and generative AI. A key source of errors in this paradigm is the imprecision of the analog hardware components. Here, we introduce a method that reduces the overall error from a linear to a quadratic dependence (from $\epsilon$ to $\epsilon^2$) on the imprecision $\epsilon$, for Gaussian sampling and linear algebra applications. The method involves sampling from an ensemble of imprecise distributions associated with various rounding events and then merging these samples. We numerically demonstrate the scalability of this method for dimensions greater than 1000. Finally, we implement this method on an actual thermodynamic computer and show $20\%$ error reduction for matrix inversion; the first thermodynamic error mitigation experiment.
\end{abstract}

\maketitle

\textit{Introduction.---}Physics-based computing has a rich history. For example, classical analog computing, which was the dominant form of computing in the early 1900s, can be viewed as physics-based, as it typically uses classical dynamical systems to solve differential equations. However, digital computers supplanted analog ones in the late 20th century due to their reliability and lack of errors. Unlike digital computing, dealing with errors is a key issue for physics-based computing.

Certainly this is the most important issue that quantum computing currently faces~\cite{preskill2018quantum}. Hence, quantum computing researchers have proposed a host of strategies to reduce the impact of errors, called error mitigation methods~\cite{endo2021hybrid,cai2023quantum}. Error mitigation is largely data-driven~\cite{lowe2020unified,bultrini2023unifying,strikis2020learning}; the basic idea is shown in Fig.~\ref{fig:ThermiesOverview}(a). For a given physical computing device, multiple inputs are sent in, the outputs are collected, and then data is post-processed (typically on a classical digital computer) to obtain a high quality result. For example, collecting data at various noise levels and extrapolating to the zero-noise limit is called Zero-Noise Extrapolation~\cite{temme2017error,kandala2018error,giurgica2020digital}, shown in Fig.~\ref{fig:ThermiesOverview}(b). Similarly, Clifford Data Regression involves collecting data for Clifford quantum circuits that are close to a target circuit and performing regression to infer the target expectation value~\cite{czarnik2020error}. In the same spirit, Probabilistic Error Cancellation takes a linear combination of data from various noisy circuits to cancel out the effect of intrinsic hardware noise~\cite{temme2017error}. These error mitigation methods are central to obtaining quantum advantage in era before fault-tolerant quantum computing~\cite{larose2022error,kim2023evidence}.

However, error mitigation has received very little attention outside the field of quantum computing. Very few error mitigation methods have been proposed for classical analog computing, as most strategies to reduce errors have been at the hardware level rather than at the data-processing level. Meanwhile, classical analog computing is recently making a comeback in light of its suitability for artificial intelligence (AI) applications~\cite{wired2023}. One of the most important areas of AI is probabilistic AI~\cite{murphy2022probabilistic}, which includes generative AI algorithms like diffusion models, Monte Carlo inference algorithms, and Bayesian neural networks. 

Recently, it was argued that thermodynamic computers are the natural match for probabilistic AI applications~\cite{coles2023thermodynamic}. Such computers are analog devices that utilize naturally occurring stochastic fluctuations to generate novel samples (e.g., for generative AI). Because noise is intentionally used a resource in this paradigm, thermodynamic computers are largely robust to (unintentional) noise like thermal noise. This is clearly distinct from standard analog computing, where noise is an issue. However, like standard analog computing, thermodynamic computers do face precision issues, where the precision of the physical components determines the accuracy of the computation.

\begin{figure}
    \centering
    \includegraphics[width=0.48\textwidth]{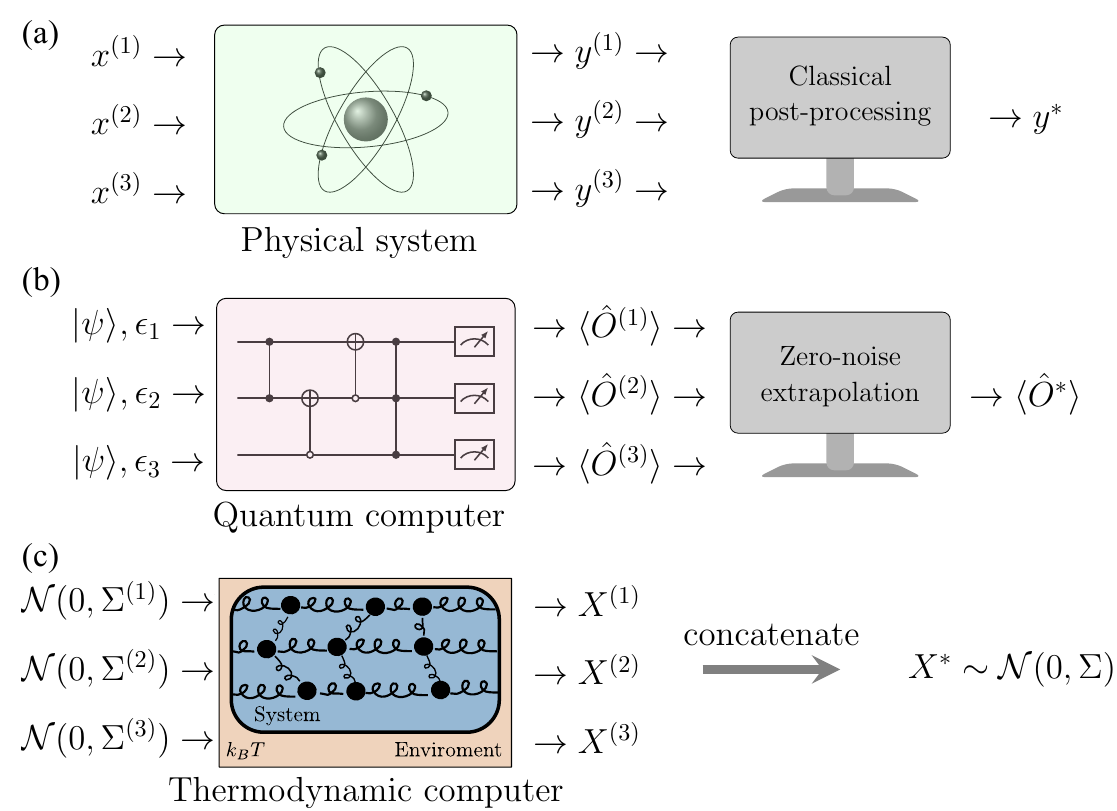}
    \caption{\textbf{Overview of Error Mitigation.} 
    (a) The general framework for error mitigation involves probing the physical computing device with multiple inputs $\{x^{(i)}\}_i$, where such inputs are typically related to the target in some way (e.g., some perturbation of the target input). The collection of outputs $\{y^{(i)}\}_i$ are then post-processed to obtain a high quality output $y^*$. (b) A plethora of error mitigation methods have been developed for quantum computers. Here we show a common one, Zero-Noise Extrapolation, where the noise level of the device is varied and the resulting expectation values $\{\langle O^{(i)}\rangle\}_i$  are fit with a curve that extrapolates to the zero noise limit, leading to the improved expectation value $\langle O^*\rangle$. (c) In thermodynamic computing, the goal is to sample from a target probability distribution $\mathcal{N}(0, \Sigma)$ with an imprecise physical device. Our error mitigation method, Thermies, randomly samples from multiple distributions $\{\mathcal{N}(0, \Sigma^{(i)})\}_i$ that are close to the target and then merges the obtained samples $\{X^{(i)}\}_i$ to obtain a high quality set of samples $X^*$ for the target distribution. }
    \label{fig:ThermiesOverview}
\end{figure}

Thermodynamic computing (TC)~\cite{conte2019thermodynamic,coles2023thermodynamic,aifer2023thermodynamic,duffield2023thermodynamic,lipka2023thermodynamic,hylton2020thermodynamic,ganesh2017thermodynamic,melanson2023thermodynamic} represents a promising, new physics-based paradigm for accelerating AI~\cite{coles2023thermodynamic} and linear algebra~\cite{aifer2023thermodynamic,duffield2023thermodynamic}. However, it remains a young field, where very little is known about its capabilities, and only recently was the first thermodynamic computer built~\cite{melanson2023thermodynamic}. In this article, we present the first error mitigation method for TC. Our method is aimed at addressing the key issue that limits the accuracy of TC, namely, imprecision. Thermodynamic computers are sampling devices, and for concreteness we focus here on the task of sampling from multivariate Gaussian distributions. Nevertheless, our analysis will have implications for sampling from non-Gaussian distributions as well as for sampling in the context of generative AI. Moreover, Gaussian sampling is a subroutine in thermodynamic linear algebra algorithms~\cite{aifer2023thermodynamic}; hence our error mitigation method is directly applicable to such algorithms. 

As noted in Fig.~\ref{fig:ThermiesOverview}(c), our method is based on merging the samples from multiple imprecise distributions associated with probabilistically either rounding up or rounding down covariance matrix elements. Our method, called Thermies (THERMIES = THermodynamic ERror Mitigation via Imprecise Ensemble Sampling), reduces the scaling from $\varepsilon$ (without error mitigation) to $\varepsilon^2$, where $\varepsilon$ quantifies the degree of imprecision. This quadratic scaling implies that small levels of imprecision do not perturb the distribution, leading to a significant gain in the accuracy of thermodynamic computers. We numerically test our method for dimensions greater than 1000, demonstrating its scalability to large dimensions. Moreover, we show that only a small number of ensemble draws (i.e., a small number of imprecise distributions) is needed to significantly reduce the error, and this number is essentially independent of dimension. Finally, we implement Thermies on a real thermodynamic computer (from Ref.~\cite{melanson2023thermodynamic}), demonstrating a performance improvement for matrix inversion on an actual device.

\bigskip

\textit{Problem Setup.---}The basic primitive that thermodynamic computers perform is sampling from some target probability distribution $p(x)$. Thinking of $p(x) \propto \exp(-U(x))$ as the exponential of a function $U(x)$, thermodynamic computers encode $U(x)$ in the energy (i.e., the Hamiltonian) of a physical system. This physical system is allowed to relax to thermal equilibrium. Then the Boltzmann distribution corresponds to the target distribution $p(x)$, and hence samples can be drawn from this Boltzmann distribution. 

When $p(x)$ is a Gaussian distribution, $U(x) = \frac{1}{2}x^{T}\Sigma^{-1} x$ is a quadratic form with $\Sigma$ being the covariance matrix~\footnote{Here we assume the mean of the distribution is zero, since the mean can be incorporated in the digital post-processing of the samples obtained from the thermodynamic computer, by adding the mean in as a displacement.}. Thus, mapping a Gaussian distribution to a thermodynamic computer involves mapping the covariance matrix $\Sigma$ to the values of hardware components, such as resistances, capacitances, or inductances in the case of electrical hardware~\footnote{In Prop.~\ref{prop2}, we extend the analysis to differentiable functions $g(\Sigma)$ of the covariance matrix, which then covers protocols or devices for which the user maps $g(\Sigma)$ to the hardware components. An important example of this is where the user maps the precision matrix $P = \Sigma^{-1}$ to the hardware, which is the protocol employed in thermodynamic matrix inversion~\cite{aifer2023thermodynamic}.}. We assume that $\Sigma$ is initially stored in a high precision form on a digital device. One then needs to upload this matrix to the thermodynamic computer. After uploading the matrix, the thermodynamic computer will typically evolve under Langevin dynamics, and hence one will employ a Langevin Monte Carlo algorithm to sample from the relevant probability distribution.

Digital devices allow up to 52 bits of precision, whereas the properties of analog components are represented with much lower precision (e.g., over the range of 3 bits to 10 bits)~\cite{stokes_inside_2007}. This leads to information loss, which impacts the distribution that one samples from on the thermodynamic computer. Perhaps the simplest approach would be to round the matrix elements in order to map the high precision matrix to a low precision matrix. This results in an erroneous matrix, for example:
\begin{equation}
    \Sigma^a = \Sigma^t + \varepsilon R,
\end{equation}
where $\Sigma^t$ is the true covariance matrix, $R$ is a matrix that represents the error due to rounding, and $\varepsilon$ quantifies the rounding error (note that $\varepsilon$ serves as a quantitative measure of the imprecision, assuming $R$ has unit norm). One can show that this simple protocol, of substituting the true matrix with a low precision surrogate, leads to a first-order correction to the probability distribution. That is, the distance of the low-precision distribution to the true distribution scales in proportion to $\varepsilon$, for small $\varepsilon$. This scaling is undesirable as it implies that even small imprecision levels will perturb the distribution that one samples from.

\bigskip

\textit{Univariate Thermies.---}Let us now introduce the Thermies error mitigation method. We begin with the one-dimensional case to explain the essential concept. Imagine a Gaussian sampling device that is capable of realizing a random variable $X$ whose probability density function is
\begin{equation}
\label{1D-precision-constraint}
    f_m(x) =\frac{1}{\sqrt{2 \pi  \varepsilon m}} \exp \left(  - \frac{1}{2} \frac{x^2}{\varepsilon m}\right), \: \: m \in \{1, 2, \dots\}.
\end{equation}
That is, the device can sample a mean-zero normally distributed random variable whose variance is a multiple of $\varepsilon$. In practice, one might require samples from a distribution having a variance that is not a multiple of $\varepsilon$. For example, suppose we would like to sample the normal distribution $\mathcal{N}[0, 1.5 \varepsilon]$, i.e., with variance $1.5 \varepsilon$. This distribution is the solid black line in Fig.~\ref{fig:Thermies1D}.

\begin{figure}
    \centering
    \includegraphics[width=0.48\textwidth]{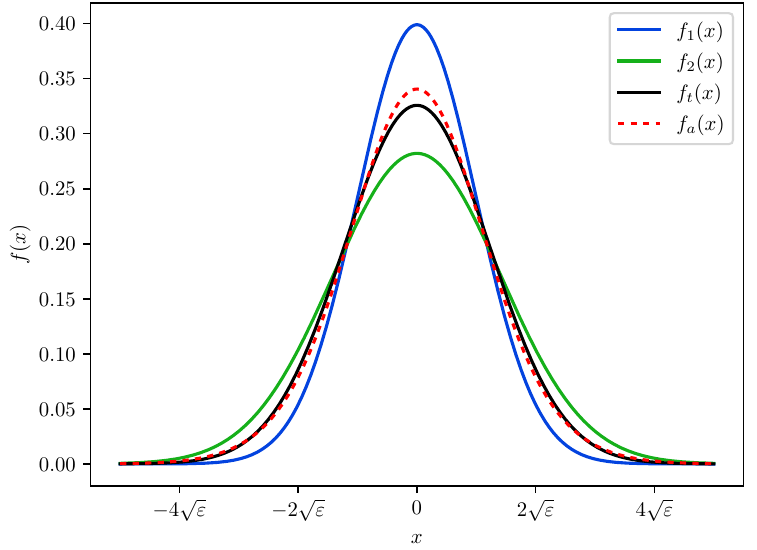}
    \caption{\textbf{Univariate Thermies.} A target distribution $f_t$ with variance $\sigma^2_t= 1.5\varepsilon$ is approximated by interpolating between $f_1$ and $f_2$ whose variances are respectively $\varepsilon$ and $2\varepsilon$. While $f_t$, $f_1$, and $f_2$ are Gaussian, $f_a$ is a Gaussian mixture.}
    \label{fig:Thermies1D}
\end{figure}

In order to (approximately) sample this distribution, we carry out the following procedure, which is the basis of the Thermies method (shown here for the special case of equal probabilities for illustration):
\begin{tcolorbox}[title={Univariate Thermies Protocol}]
\begin{enumerate}
  \item Sample a Bernoulli random variable $B \in \{ 0,1\}$ with $\text{Pr}(B=0)=1/2$ and $\text{Pr}(B=1)=1/2$.
  \item If the outcome is $B=0$, sample $\mathcal{N}[0, \varepsilon]$, and if the outcome is $B=1$, sample $\mathcal{N}[0,2\varepsilon]$.
  \item Record the result as a realization of random variable $X$, without storing the outcome of the Bernoulli trial.
\end{enumerate}
\end{tcolorbox}
The probability density function of the random variable $X$ is then
\begin{equation}
\label{1D-approx-pdf}
f_a(x) = \frac{1}{2}f_1(x) + \frac{1}{2} f_2(x),
\end{equation}
where the subscript $a$ signifies that this distribution is an approximation to a target distribution $f_t$. As the distributions $f_1$ and $f_2$ are both zero mean, $f_a$ is zero mean as well. Moreover, the variance of $f_a$ is 
\begin{align}
    \sigma^2_a &= \int_{-\infty}^\infty dx f_a(x) x^2\\
    &=\frac{1}{2}\int_{-\infty}^\infty dx f_1(x) x^2 + \frac{1}{2}\int_{-\infty}^\infty dx f_2(x) x^2 \\
    & = \frac{1}{2}(\varepsilon) + \frac{1}{2}(2\varepsilon) = 1.5 \varepsilon.
\end{align}
For this example, the $f_a$ distribution is plotted in Fig.~\ref{fig:Thermies1D} as the red dashed line, which matches the target distribution better than either $f_1$ and $f_2$ do.

It is straightforward to see that if we instead form a mixture $f_a = (1-w) f_m + w f_{m+1}$ with $w \in [0,1]$, the variance of $f_a$ is
\begin{equation}
\label{1D-approx-sigma}
\sigma^2_a  = (m+w)\varepsilon,
\end{equation}
which means that we can choose $w$ and $m$ to obtain arbitrary $\sigma^2_a \geq\varepsilon$, so we can ensure that $\sigma^2_a = \sigma^2_t$, which we refer to as variance-matching. Obtaining an approximating distribution with $\sigma^2_a <\varepsilon$ is problematic, an issue which we will revisit later~\footnote{We remark that in the univariate case, the imprecision problem could be avoided by rescaling the random variable to have a realizable variance. That is, we could always define a new random variable $Y \propto X$ such that $Y \sim \mathcal{N}[0, m\varepsilon]$ for some $m \in \{ 1, 2, \dots\}$, so no error mitigation would be necessary. This is generally impossible in the multivariate case, which is why error mitigation is necessary.}.

\bigskip

\textit{Multivariate Thermies.---}We now consider a thermodynamic computer that can sample a $d$-dimensional zero-mean multivariate normal distribution:
\begin{equation}
    \label{multivariate-normal-pdf}
    f_{0;\Sigma}(x) = \frac{1}{(2\pi)^{d/2}\left|  \Sigma \right|^{1/2}} \exp \left( - \frac{1}{2} x^\intercal  \Sigma^{-1} x \right),
\end{equation}
whose covariance matrix has elements that are multiples of the imprecision parameter $\varepsilon$. That is, one can sample any distribution $\mathcal{N}[0, \varepsilon \tilde{\Sigma}]$ with
\begin{equation}
\label{ND-precision-constraint}
    \tilde{\Sigma} \in \text{PSD}_d(\mathbb{Z}),
\end{equation}
with $\text{PSD}_d(\mathbb{Z})$ the set of integer positive semi-definite $d \times d$ matrices. Equation~\eqref{ND-precision-constraint} is the multidimensional generalization of the constraint in Eq.~\eqref{1D-precision-constraint}. To extend Thermies to the multidimensional setting, we would like to match the covariance matrix of a target normal distribution $\mathcal{N}[0, \Sigma^t]$ using an ensemble of realizable distributions $\mathcal{N}[0, \varepsilon \tilde{\Sigma}]$ where $\tilde{\Sigma} \in \text{PSD}_d(\mathbb{Z})$. We now give a procedure for achieving this:
\begin{tcolorbox}[title={Multivariate Thermies Protocol},breakable]
\begin{enumerate}
    \item Compute the residual matrix $R = \Sigma^t/\varepsilon -   \floor{\Sigma^t/\varepsilon}$.
    \item For each pair $(i, j) \in \{1,2\dots d\}^2$ with $i \leq j$, draw a realization of the Bernoulli random variable $B_{ij}$ which has probabilities $\text{Pr}(B_{ij}=0) = 1-R_{ij}$ and $\text{Pr}(B_{ij}=1) = R_{ij}$. Also define $B_{ji} = B_{ij}$, resulting in a matrix of realizations $B \in \{0,1\}^{d \times d}$.
    \item Construct the matrix $\Sigma^B = \varepsilon (\floor{\Sigma^t/\varepsilon} +  B)$, and draw a sample from the normal distribution $\mathcal{N}[0, \Sigma^B]$.
    \item Record the resulting sample as a realization of the random vector $X$, without storing the results of the Bernoulli trials.
\end{enumerate}
\end{tcolorbox}
The procedure just described differs from the univariate procedure in requiring many Bernoulli random variables to be sampled for each error mitigated sample of the vector $X$, whereas a single Bernoulli random variable was needed in the univariate case. Let $D=(d^2 + d)/2$ be the number of Bernoulli random variables drawn. If we vectorize the matrix $B$, then $\vec{B}$ belongs to the $D$-dimensional hypercube's vertex set $\{0,1\}^{D}$. (See Fig.~\ref{fig:hypercube_a} for an example of this hypercube.) There are therefore $2^D$ possible outcomes of the set of Bernoulli trials, resulting in $2^D$ possible covariance matrices $\Sigma^B$. The resulting sample of $X$ which is obtained can then be seen as arising from a Gaussian mixture,
\begin{equation}
\label{ND-approx-pdf}
    f_a = \sum_{\vec{b} \in \{0,1\}^D} w_bf_{0;\Sigma^{b}},
\end{equation}
where the parameters $w_b$ (which we call the weights) are joint probabilities, $w_b = \text{Pr}(B = b)$. 
%\begin{equation}
%    w_b = \text{Pr}(B = b).
%\end{equation}
As all of the components of the mixture are mean zero, $\braket{X}=0$, and we again see that the covariance of the mixture is given by the corresponding convex combination of the components' covariances,
\begin{align}
    \Sigma^a_{ij} &= \int_{-\infty}^\infty d^dx f_a(x) x_i x_j\\
    &=\sum_{\vec{b} \in \{0,1\}^D} w_b\int_{-\infty}^\infty d^dx f_{0;\Sigma^{b}}(x) x_i x_j \\
    & = \sum_{\vec{b} \in \{0,1\}^D} w_b\Sigma^b_{ij}.
\end{align}
We may then write the above as $\Sigma^a = \braket{\Sigma^B}$, where the expected value is understood as being taken with respect to the weights $w_b$. Here we note that, by construction, each matrix $\Sigma^b$ is of the form $\varepsilon \tilde{\Sigma}^b$, with $\tilde{\Sigma}^b \in \text{Sym}_d(\mathbb{Z})$. This means that the elements of $\Sigma^b$ can be encoded in imprecise hardware. However, in order for these matrices to represent physically realizable distributions, we must also have $\tilde{\Sigma}^b \in \text{PSD}_d(\mathbb{Z})$. This is not guaranteed for arbitrary target distributions, an issue which we return to later. In what follows, we refer to the matrices $\Sigma^b$ as the nearest neighbors of the target covariance matrix $\Sigma^t$, and similarly the associated normal distributions $f_{0;\Sigma^b}$ are nearest neighbors of the target distribution $f_{0;\Sigma^t}$.

\begin{figure}%
\centering
\subfloat[][]{%
\label{fig:hypercube_a}%
\centering
\includegraphics[width=0.95\linewidth, height=0.65\linewidth]{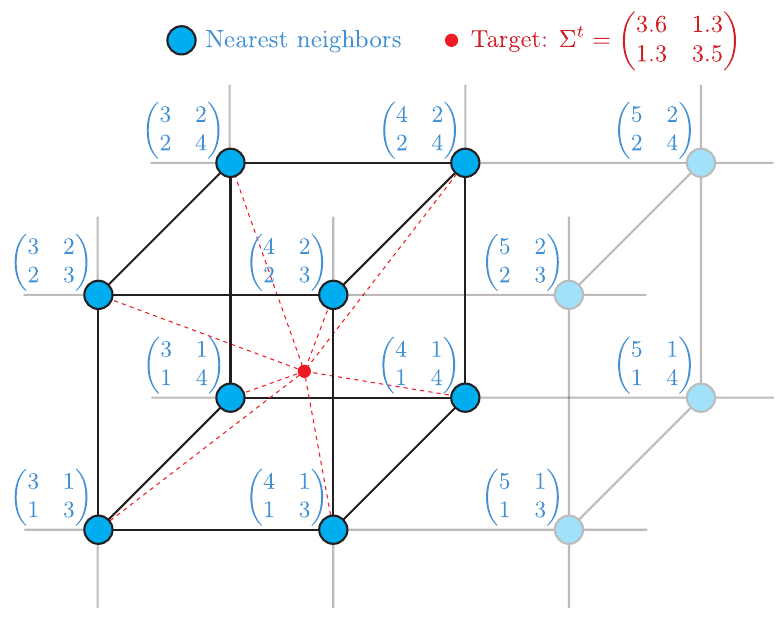}}%
\hspace{1pt}%
\subfloat[][]{%
\label{fig:2d_ellipse_b}%
\includegraphics[width=\linewidth, height=0.8\linewidth]{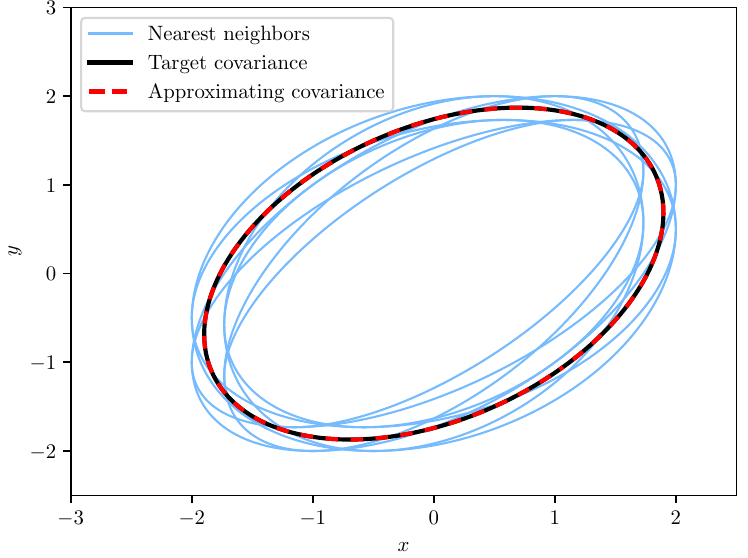}}
\caption[]{\textbf{Two-Dimensional Example.} (a) The hypercube (in this case a cube) representation of the nearest neighbor ensemble is shown for a two-dimensional target covariance matrix $\Sigma^t$. The red dot is $\Sigma^t$ and the blue dots are (imprecise) approximations of the target. The target is inside the convex hull of its nearest neighbors, so it can be expressed as a convex combination of them. (b) The covariance ellipses are shown for the target covariance (black), the $8$ nearest neighbor covariances (blue), and the convex combination of the nearest neighbor covariances (red dashed). The target and approximating covariances exactly coincide due to the covariance matching property of Thermies.}%
\label{fig:2d_example_hypercube}%
\end{figure}

As shown in the Supplemental Material, Thermies results in an approximate distribution that satisfies a covariance matching condition:
\begin{equation}
\label{ND-covariance-matching}
    \Sigma^a = \Sigma^t.
\end{equation}
For illustration, we show a two-dimensional example of Thermies in Figure~\ref{fig:2d_example_hypercube}. Figure~\ref{fig:hypercube_a} displays the hypercube (in this case a cube) of the nearest neighbor ensemble, while Fig.~\ref{fig:2d_ellipse_b} shows the covariance ellipses for this ensemble and illustrates the idea of covariance matching. The details of this example are given in the Supplemental Material.

\bigskip
\textit{Imprecision Dependence.---}It is desirable for the quality of samples to be insensitive to the precision of the hardware implementation. Fortunately, the Thermies protocol eliminates the first order dependence of the approximate distribution on $\varepsilon$ as $\varepsilon$ goes to zero. We formally state this as follows.

\begin{proposition}
\label{epsilon-independence-prop}
Given some $\Sigma^t \in \text{PSD}_d(\mathbb{R})$, suppose that there exist weights $w_1 \dots w_N \in \mathbb{R}^{\geq 0}$ and matrices  $\Sigma^1 \dots \Sigma^N \in \text{PSD}_d(\mathbb{R})$ such that $\sum_i w_i = 1$ and $\sum_i w_i \Sigma^i = \Sigma^t$. Then define 
\begin{equation}
    P^i = \frac{1}{\varepsilon}(\Sigma^i - \Sigma^t),
\end{equation}
so  $\Sigma^b = \Sigma^t + \varepsilon P^b$. Define the function
\begin{equation}
    f_a(x) = \sum_b w_b N_b \exp\left(- \frac{1}{2} x^\intercal(\Sigma^t+\varepsilon P^b)^{-1}x\right),
\end{equation}
where $N_b = \left[(2 \pi)^{d/2}\sqrt{|\Sigma^t + \varepsilon P^b|}\right]^{-1}$. Then
\begin{equation}
    \lim_{\varepsilon' \to 0}\left.\partial_\varepsilon f_a(x;\varepsilon)\right|_{\varepsilon'} = 0.
\end{equation}
\end{proposition}

A proof of Prop.~\ref{epsilon-independence-prop} is given in the Supplemental Material. It follows immediately from Prop.~\ref{epsilon-independence-prop} that the Lebesgue $L_1$ norm of the difference between the approximate distribution $f_a$ and the target distribution $f_t$ has vanishing first-order dependence on $\varepsilon$ as $\varepsilon$ goes to zero. However, this is not the case if the Thermies protocol is not used. Figure \ref{fig:LinearQuadratic} shows the $L_\infty$ distance between $f_a$ and $f_t$ with and without Thermies error mitigation. While in the error-mitigated case, the slope vanishes at $\varepsilon = 0$, evidently the slope does not vanish at $\varepsilon = 0$ in the unmitigated case. This numerical result provides evidence that Thermies (or a similar method) is necessary to eliminate sensitivity of sample quality to hardware imprecision.

\bigskip

\textit{Extension to differentiable functions.---}As stated in the previous section, our error mitigation protocol eliminates the first-order dependence of the probability density on $\varepsilon$ as $\varepsilon$ goes to zero. In fact, a more general result can be shown, which states that the first-order dependence on $\varepsilon$ is also eliminated for arbitrary differentiable functions of the covariance matrix.

\begin{proposition}\label{prop2}
    Suppose $g:\text{PSD}_d(\mathbb{R}) \to \mathbb{R}$ is differentiable at $\Sigma = \Sigma^t$, and define $\hat{g} = \sum_i w_i g(\Sigma^i)$. Then $\hat{g} = g(\Sigma^t)$ when $\varepsilon = 0$ and 
    \begin{equation}
         \lim_{\varepsilon'\to 0  }\left.\frac{\partial \hat{g}}{\partial \varepsilon}\right|_{\varepsilon'}
        = 0.
    \end{equation}
\end{proposition}
\begin{proof}
To see this explicitly, we compute
\begin{align}
    \frac{\partial \hat{g}}{\partial \varepsilon}
    &= \frac{\partial}{\partial \varepsilon} \sum_i w_i g\left(\Sigma^i \right)\\
    &= \sum_{ijk}  w_i\left. \frac{\partial g(\Sigma)}{\partial \Sigma_{jk}}\right|_{\Sigma^i}\frac{\partial \Sigma^i_{jk}}{\partial \varepsilon}
\end{align}
When $\varepsilon = 0$, $\Sigma^i = \Sigma^t$ for all $i$, so
\begin{align}
   \left. \frac{\partial \hat{g}}{\partial \varepsilon}\right|_{\varepsilon = 0}
    &= \sum_{ijk}  w_i\left. \frac{\partial g(\Sigma)}{\partial \Sigma_{jk}}\right|_{\Sigma^t}\frac{\partial \Sigma^i_{jk}}{\partial \varepsilon} \\
    &= \sum_{jk} \left. \frac{\partial g(\Sigma)}{\partial \Sigma_{jk}}\right|_{\Sigma^t}\frac{\partial }{\partial \varepsilon}\left( \sum_{i}w_i \Sigma^i_{jk}\right) \\
    &=0,
\end{align}
where the last line is achieved using the fact that $\sum_i w_i \Sigma^i = \Sigma^t$, by definition. Clearly when $\varepsilon=0$ we have $\hat{g} = g(\Sigma^t)$.
\end{proof}

The implication of this simple calculation is that $\hat{g}$ may be used as an estimate for any differentiable function $g$. Note that $g$ does not have to be differentiable everywhere, it is sufficient that $g$ is differentiable at $\Sigma^t$. The Thermies method can therefore be used to eliminate first-order dependence on imprecision in a variety of thermodynamic algorithms; these include, for example, solving a linear system of equations, inverting a matrix, solving the Lyapunov equation, estimating of the determinant of a matrix, and exponentiating a matrix~\cite{aifer2023thermodynamic,duffield2023thermodynamic}. Also note that Prop.~\ref{epsilon-independence-prop} is implied as a special case of Prop.~\ref{prop2} (hence giving an alternative proof of Prop.~\ref{epsilon-independence-prop}), as the probability density function is differentiable with respect to elements of the covariance matrix.

\begin{figure}
    \centering
    \includegraphics[width=0.48\textwidth]{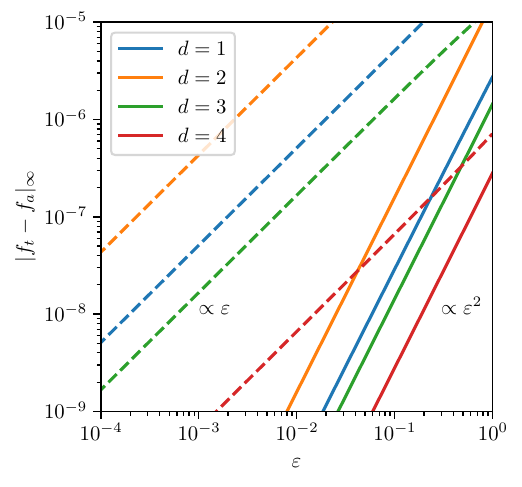}
    \caption{\textbf{Imprecision Dependence}. The $L_\infty$ distance between the approximate distribution $f_a$ and the target $f_t$ is plotted versus $\varepsilon$ for dimensions $1$ to $4$ with error mitigation (solid lines) and without error mitigation (dashed lines). With error mitigation, the error is proportional to $\varepsilon^2$ near $\varepsilon = 0$, so the first order dependence is zero as $\varepsilon$ goes to zero. Without error mitigation the error is proportional to $\varepsilon$ for small $\varepsilon$, so the error is sensitive to imprecision.}
    \label{fig:LinearQuadratic}
\end{figure}

\begin{figure}[t]
    \centering
    \includegraphics[width=0.48\textwidth]{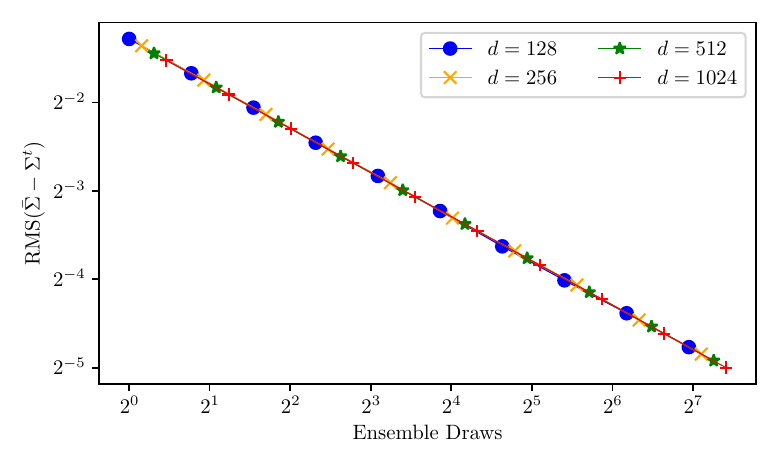}
    \caption{\textbf{Sampling Complexity.} The Root Mean Square (RMS) error between $\Sigma^t$ and the average of $M$ covariance matrices drawn from the nearest neighbor ensemble is plotted versus $M$. The RMS error shows a dependence of roughly $M^{-1/2}$. This behavior is independent of $d$.}
    \label{fig:ConvergenceAvgCov}
\end{figure}

\bigskip
\textit{Sampling Complexity.---}The Thermies protocol described above has the advantage of eliminating the first-order dependence of the approximate distribution  on the imprecision parameter, and moreover allows for perfectly matching the covariance matrix of the target distribution. However, the method has the disadvantage that one must draw a covariance matrix from the ensemble of nearest neighbors for each new sample which is required. Naturally, the question arises of whether high-quality samples can be obtained with fewer draws from the nearest neighbor ensemble. To address this, we construct the Thermies protocol with repetition (see Supplemental Material for elaboration), which requires fewer draws from the ensemble of nearest neighbors, and analyze its performance. The only difference between the Thermies method given before and this new one is that now $n$ samples are drawn from the normal distribution $\mathcal{N}[0,\Sigma^B]$ instead of only $1$. If the Thermies protocol with repetition is carried out $M$ times consecutively, then we end up with a sequence of $N = Mn$ realizations of $X$. Given that the nearest neighbor ensemble may be very large (in fact comprising $2^{(d^2 + d)/2}$ elements), one might be led to wonder whether a relatively very small selection of $M$ elements can provide a representative sample of this ensemble. In fact it can, which simply follows from Hoeffding's inequality, as we show in the Supplemental Material. In particular, let $\bar{\Sigma}$ be the sample mean of a sequence of covariance matrices $\Sigma^1 \dots \Sigma^M$ drawn from the nearest-neighbor ensemble of a target covariance matrix $\Sigma^t$:
\begin{equation}
\label{sigma-bar-def}
\bar{\Sigma} = \frac{1}{M}\sum_{m=1}^M \Sigma^m.
\end{equation}
Choosing some $\delta >0$, Hoeffding's inequality says \cite{hoeffding1948class, hoeffding1963probability}
\begin{equation}
\label{hoeffding-bound}
    \text{Pr}\left( |\bar{\Sigma}_{ij} - \Sigma^t_{ij}| \geq \delta \right) \leq 2 \exp\left(-2 \frac{M \delta^2}{ \varepsilon^2} \right).
\end{equation}

If we hold the left side of \eqref{hoeffding-bound} constant and vary $M$ (while also holding $\varepsilon$ constant), we see that $\delta \propto M^{-1/2}$, which implies that the elementwise difference between $\bar{\Sigma}$ and $\Sigma^t$ goes with $M^{-1/2}$. Importantly, as there are no factors in \eqref{hoeffding-bound} which involve $d$, this behavior is independent of dimension. In fact this is exactly what we see when random realizations of the nearest neighbor ensemble are realized and $\bar{\Sigma}$ is evaluated, as Fig. \ref{fig:ConvergenceAvgCov} shows. We also see that a relatively small number of samples from the nearest neighbor ensemble is needed to get a good estimate of $\Sigma^t$.

\bigskip

\textit{Feasibility for Ill-conditioned Matrices.---}In the Supplemental Material, we delve into the question of how the condition number $\kappa$ of the target covariance matrix impacts the feasibility of applying the Thermies method. Namely, we derive a fundamental relation between $\kappa$, the dimension $d$, and the imprecision $\varepsilon$, and this relation provides the region of ``space'' (space here refers to the coordinates ($\kappa$, $d$, $\varepsilon$)) in which Thermies can be successfully applied. Qualitatively, coordinates ($\kappa$, $d$, $\varepsilon$) that are closer to the origin are more amenable to the feasibility of Thermies, although we provide the precise quantitative relationship in the Supplemental Material.

\begin{figure}[t]
    \centering    \includegraphics[width=\linewidth, height=0.8\linewidth]{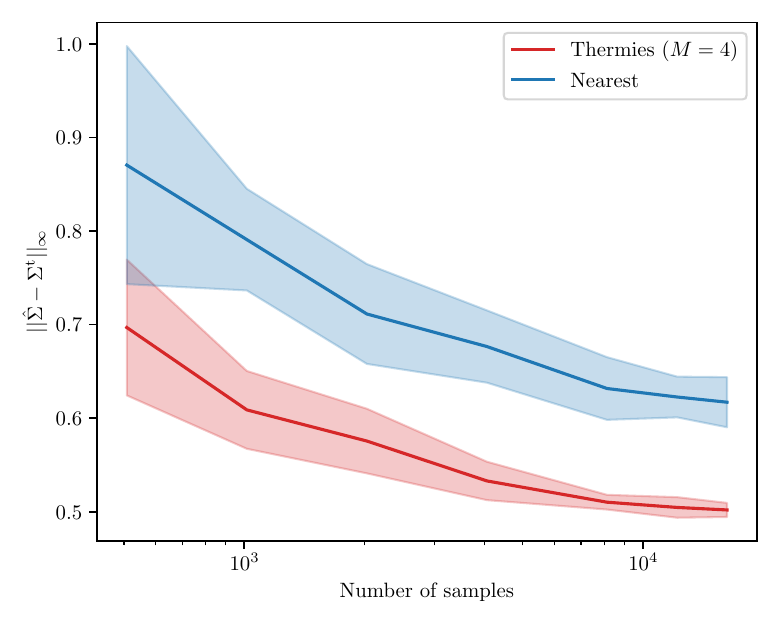}
\caption{\textbf{Implementation.} We show the results for implementing Thermies with 4 ensemble draws on a thermodynamic computer for the task of 8x8 matrix inversion. The solid lines represent the mean of the results from 10 repetitions while the shaded regions represent the standard deviation. The final error is reduced by $\sim20\%$ with Thermies (red curve) relative to the case of no error mitigation (blue curve).}
\label{fig:PCB_Implementation}
\end{figure}

\bigskip

\textit{Implementation.---}As a proof of concept, we implemented Thermies on a real thermodynamic computer, namely, the computer presented in Ref.~\cite{melanson2023thermodynamic}, which is composed of 8 fully connected unit cells on a printed circuit board. The results are shown in Fig.~\ref{fig:PCB_Implementation}. Here we employ Gaussian sampling as a subroutine for matrix inversion (see Ref.~\cite{aifer2023thermodynamic} for the thermodynamic matrix inversion algorithm). Hence, the task shown in Fig.~\ref{fig:PCB_Implementation} is inverting an 8x8 matrix. The matrix inversion error naturally goes down with the number of samples gathered. However, the final error is significantly less ($\sim20\%$ error reduction) for the case where Thermies is applied (red curve), versus the case with no error mitigation (blue curve). We emphasize that the device from Ref.~\cite{melanson2023thermodynamic} is relatively low precision, and even better performance for Thermies is expected in the high-precision limit. More details on this implementation are presented in the Supplemental Material.

\bigskip

\textit{Conclusion.---}We have presented the first error mitigation method for Thermodynamic Computing. We targeted the dominant source of error in the paradigm: imprecision of hardware components. Without error mitigation, the distribution error grows linearly with imprecision $\varepsilon$, whereas the distribution is unaffected (i.e., has zero derivative) in the small $\varepsilon$ limit when employing our method. The convergence of the sample-average covariance matrix $\bar{\Sigma}$ is dimension independent, demonstrating the potential for scalability to large dimensions. A natural future direction will be to extend this method to non-Gaussian sampling.

\bibliography{thermo.bib}

\newpage
\pagebreak
\onecolumngrid

\begin{appendix}

\begin{center}
\Large Supplemental Material for\\ ``Error Mitigation for Thermodynamic Computing''
\end{center}

\section{Overview}

In this Supplemental Material, we discuss and elaborate on the following topics: 

\begin{itemize}
    \item Analysis of the Multivarate Thermies Method
    \item Details on the Two-dimensional Example
    \item Proof of Proposition~\ref{epsilon-independence-prop}
    \item Thermies Protocol with Repetition
    \item Convergence of Thermies method via Hoeffding's bound
    \item Implementation on thermodynamic hardware
\end{itemize}

\section{Analysis of the Multivarate Thermies Method}
%\label{app:epsilon-independence}

Here we give a more detailed analysis of the multivariate Thermies method. It is helpful to first introduce the following notation for various sets of matrices:
\begin{itemize}
    \item $\text{Sym}_d(\mathbb{R})$  is the set of real symmetric $d \times d$ matrices. We denote by $\text{Sym}_d(\mathbb{Z})$ and $\text{Sym}_d(\{0,1\})$ respectively the sets of symmetric matrices with integer and binary entries.

    \item A matrix $A \in \text{Sym}_d(\mathbb{R})$ may be ``vectorized'', or treated as a vector $\vec{A} \in \mathbb{R}^{(d^2 + d)/2}$. The vector is formed by appending the rows of the upper triangle of $A$ together in order, i.e.
    \begin{equation}
        \vec{A} = (A_{11}, A_{12}, \dots A_{1d}, A_{22}, A_{23},\dots A_{2d}, \dots)^\intercal.
    \end{equation}
    As $\text{Sym}_d(\{0,1\}) \subset \text{Sym}_d(\mathbb{Z})\subset \text{Sym}_d(\mathbb{R})$, the same notation may be used for integer and binary symmetric matrices.

    \item $\text{PSD}_d(\mathbb{R})$ is the set of real positive semi-definite $d \times d$ matrices, defined as \cite{bellman_introduction_1997}

\begin{equation}
\label{PSD-def}
\text{PSD}_d(\mathbb{R}) = \{A \in \text{Sym}_d(\mathbb{R}) : x^\intercal A x  \geq 0 \forall x \in \mathbb{R}^d \setminus 0\},
\end{equation}
and similarly, $\text{PSD}_d(\mathbb{Z}) =\text{PSD}_d(\mathbb{R}) \cap \mathbb{Z}^{d \times d}$.
\end{itemize}

A $d$-dimensional multivariate normal distribution is described by a mean vector $\mu \in \mathbb{R}^d$ and a covariance matrix $\Sigma \in \text{PSD}_d(\mathbb{R})$. If $X$ is a normally distributed random vector we write $X \sim \mathcal{N}[\mu, \Sigma]$, and denote the associated probability density function by $f_{\mu;\Sigma}$, which is given by \cite{kac_characterization_1939}
\begin{equation}
    \label{multivariate-normal-pdf_SM}
    f_{\mu;\Sigma}(x) = \frac{1}{(2\pi)^{d/2}\left|  \Sigma \right|^{1/2}} \exp \left( - \frac{1}{2} (x - \mu)^\intercal  \Sigma^{-1} (x-\mu) \right).
\end{equation}

We now suppose that there is a device (i.e., a thermodynamic computer) which can sample a $d$-dimensional zero-mean multivariate normal distribution whose covariance matrix has elements that are multiples of the imprecision parameter $\varepsilon$. That is, we are able to sample any distribution $\mathcal{N}[0, \varepsilon \tilde{\Sigma}]$ with
\begin{equation}
\label{ND-precision-constraint_App}
    \tilde{\Sigma} \in \text{PSD}_d(\mathbb{Z}).
\end{equation}
Equation \eqref{ND-precision-constraint_App} can be viewed as the multidimensional generalization of the constraint given in Eq. \eqref{1D-approx-pdf}. The goal of the multivariate Thermies method is to match the covariance matrix of a target normal distribution $\mathcal{N}[0, \Sigma^t]$ using an ensemble of realizable distributions $\mathcal{N}[0, \varepsilon \tilde{\Sigma}]$ where $\tilde{\Sigma} \in \text{PSD}_d(\mathbb{Z})$.

In the main text, we presented the multivariate Thermies protocol. Recall from the main text that this protocol requires many Bernoulli random variables to be sampled for each error mitigated sample of the vector $X$. Let $D=(d^2 + d)/2$ be the number of Bernoulli random variables drawn. If we vectorize the matrix $B$, then $\vec{B}$ belongs to the $D$-dimensional hypercube's vertex set $\{0,1\}^{D}$. There are therefore $2^D$ possible outcomes of the set of Bernoulli trials, resulting in $2^D$ possible covariance matrices $\Sigma^B$. The resulting sample of $X$ which is obtained can then be seen as arising from a Gaussian mixture,
\begin{equation}
\label{ND-approx-pdf_SM}
    f_a = \sum_{\vec{b} \in \{0,1\}^D} w_bf_{0;\Sigma^{b}},
\end{equation}
where the weights $w_b$ are joint probabilities, $w_b = \text{Pr}(B = b)$. As all of the components of the mixture are mean zero, $\braket{X}=0$, and we again see that the covariance of the mixture is given by the corresponding convex combination of the components' covariances,
\begin{align}
%    \Sigma^a_{ij} &= \int_{-\infty}^\infty d^dx f_a(x) x_i x_j\\
%    &=\sum_{\vec{b} \in \{0,1\}^D} w_b\int_{-\infty}^\infty d^dx f_{0;\Sigma^{b}}(x) x_i x_j \\
    \Sigma^a_{ij} = \sum_{\vec{b} \in \{0,1\}^D} w_b\Sigma^b_{ij}.
\end{align}
We may then write $\Sigma^a = \braket{\Sigma^B}$, where the expected value is understood as being taken with respect to the weights $w_b$. 

%Here we note that, by construction, each matrix $\Sigma^b$ is of the form $\varepsilon \tilde{\Sigma}^b$, with $\tilde{\Sigma}^b \in \text{Sym}_d(\mathbb{Z})$. This means that the elements of $\Sigma^b$ can be encoded in imprecise hardware. However, in order for these matrices to represent physically realizable distributions, we must also have $\tilde{\Sigma}^b \in \text{PSD}_d(\mathbb{Z})$. This is not guaranteed for arbitrary target distributions, an issue which we return to later. 

Continuing the analysis of the Thermies protocol, we notice that each matrix element $\Sigma^B_{ij}$ depends only on a single Bernoulli random variable $B_{ij}$, and therefore the elements of $\Sigma^B_{ij}$ are independent random variables, whose probability distributions are
\begin{align}
    &\text{Pr}\left(\Sigma^B_{ij} = 
    \varepsilon\floor{\Sigma^t_{ij}/\varepsilon}\right) = 1-R_{ij} 
    , \\
    &\text{Pr}\left(\Sigma^B_{ij} = 
    \varepsilon \floor{\Sigma^t_{ij}/\varepsilon} + \varepsilon\right) = R_{ij}.
\end{align}
The expected value of each matrix element is
\begin{align}
\braket{\Sigma^B_{ij}} &= \text{Pr}\left(\Sigma^B_{ij} = 
    \varepsilon \floor{\Sigma^t_{ij}/\varepsilon}\right)\left(\varepsilon \floor{\Sigma^t_{ij}/\varepsilon}\right) + 
 \text{Pr}\left(\Sigma^B_{ij} = 
    \varepsilon \floor{\Sigma^t_{ij}/\varepsilon} + \varepsilon\right)\left(\varepsilon  \floor{\Sigma^t_{ij}/\varepsilon} + \varepsilon \right) \\
    & = (1-R_{ij})\left(\varepsilon \floor{\Sigma^t_{ij}/\varepsilon}\right) + R_{ij}\left(\varepsilon  \floor{\Sigma^t_{ij}/\varepsilon} + \varepsilon \right) \\
    &= (1+\floor{\Sigma^t_{ij}/\varepsilon} - \Sigma^t_{ij}/\varepsilon)\left(\varepsilon \floor{\Sigma^t_{ij}/\varepsilon}\right) + \left( \Sigma^t_{ij}/\varepsilon - \floor{\Sigma^t_{ij}/\varepsilon}\right)\left(\varepsilon  \floor{\Sigma^t_{ij}/\varepsilon} + \varepsilon \right) \\
    & = \Sigma^t_{ij}.
\end{align}
Therefore the approximate distribution satisfies a covariance-matching condition,
\begin{equation}
\label{ND-covariance-matching_SM}
    \Sigma^a = \Sigma^t.
\end{equation}

We may also write the average $\braket{\Sigma^B}$ by summing over the joint distribution for all elements of $B$ simultaneously. Each weight $w_b$ can be expressed as product over matrix elements as
\begin{equation}
\label{interpolation-weight-formula}
    w_b= \prod_{i=1}^d \prod_{j\geq i}^d  \begin{cases} 1-R_{ij} & b_{ij} = 0 \\ 
    R_{ij} & b_{ij} = 1\end{cases},
\end{equation}
or in vectorized notation as
\begin{equation}
\label{interpolation-weight-formula-vec}
    w_b= \prod_{i=1}^D  \begin{cases} 1-\vec{R}_i & \vec{b}_i = 0 \\ 
    \vec{R}_i & \vec{b}_i = 1\end{cases}.
\end{equation}
As the weights are probabilities, they satisfy $w_b \geq 0$ and $\sum_b w_b=1$, meaning $\Sigma^a$ is a convex combination of the matrices $\Sigma^b$ for $\vec{b} \in \{0,1\}^D$. Note that because $\Sigma^B$ is related to $B$ by a scaling and translation, the vectors $\vec{\Sigma}^b$ comprise the vertex set of a hypercube as well. The convex hull of the vertex set of a hypercube is the hypercube itself \cite{klee_what_1971}, meaning that for the appropriate choice of weights $w_b$, any point may be obtained within a $D$-dimensional hypercube
\begin{equation}
    \sum_{\vec{b}\in\{0,1\}^D} w_b \Sigma^b \in \prod_{i=1}^d \prod_{j\geq i}^d  [\floor{\Sigma^t_{ij}/\varepsilon}, \floor{\Sigma^t_{ij}/\varepsilon}+\varepsilon].
\end{equation}
The particular choice of Eq. \eqref{interpolation-weight-formula} corresponds to the well known multilinear interpolation formula \cite{zhang_fast_2021}, and results in the covariance matching condition \eqref{ND-covariance-matching_SM}. With this picture in mind, we refer to the matrices $\Sigma^b$ as the nearest neighbors of the target covariance matrix $\Sigma^t$, and similarly the associated normal distributions $f_{0;\Sigma^b}$ are nearest neighbors of the target distribution $f_{0;\Sigma^t}$.

\subsection{Details on the Two-dimensional Example}

For additional clarity, we include an example of the Thermies method for the case $d=2$, as illustrated in Fig.~\ref{fig:2d_example_hypercube}. Specifically, we describe the use of the protocol elaborated in the previous section to sample the normal distribution $\mathcal{N}[0,\Sigma^t]$, with covariance matrix
\begin{equation}
    \label{2D-example-sigma}
    \Sigma^t = \begin{pmatrix}3.6 & 1.3 \\ 1.3 & 3.5\end{pmatrix}.
\end{equation}
We set $\varepsilon = 1$, so
\begin{equation}
    \varepsilon\floor{\Sigma/\varepsilon} = \begin{pmatrix}3 & 1 \\ 1 & 3\end{pmatrix}.
\end{equation}

Using the vectorized notation, we write the nearest neighbor covariance matrices $\Sigma^{\vec{b}}$ as
\begin{align}
\label{2D-neighbors}
    \Sigma^{000} &= \begin{pmatrix}3 & 1 \\ 1 & 3\end{pmatrix}, \: \: 
    \Sigma^{001} = \begin{pmatrix}3 & 1 \\ 1 & 4\end{pmatrix},\: \: 
    \Sigma^{010} = \begin{pmatrix}3 & 2 \\ 2 & 3\end{pmatrix}, \: \: 
    \Sigma^{011} = \begin{pmatrix}3 & 2 \\ 2 & 4\end{pmatrix},\\
    \Sigma^{100} &= \begin{pmatrix}4 & 1 \\ 1 & 3\end{pmatrix}, \: \: 
    \Sigma^{101} = \begin{pmatrix}4 & 1 \\ 1 & 4\end{pmatrix},\: \: 
    \Sigma^{110} = \begin{pmatrix}4 & 2 \\ 2 & 3\end{pmatrix}, \: \: 
    \Sigma^{111} = \begin{pmatrix}4 & 2 \\ 2 & 4\end{pmatrix}.
\end{align}
These matrices are portrayed graphically as the vertices of a hypercube, which is a cube in this case, in Figure~\ref{fig:hypercube_a}. Evaluating the weights using Eq. \eqref{interpolation-weight-formula-vec}, we find
\begin{align}
\label{2D-weights}
    w_{000} &= 0.14, \: \: w_{001} = 0.14, \: \: w_{010} = 0.06, \: \: w_{011} = 0.06\\
    w_{100} &= 0.21, \: \: w_{101} = 0.21, \: \: w_{110} = 0.09, \: \: w_{111} = 0.09.
\end{align}
It is easily verified that the covariance matching condition is satisfied, that is
\begin{equation}
    \sum_{\vec{b}\in\{0,1\}^3} w_{\vec{b}}\Sigma^{\vec{b}} = \Sigma^t.
\end{equation}

\section{Feasibility for Ill-conditioned Matrices}

\begin{figure}[t]
    \centering
    \includegraphics[width=0.5\textwidth]{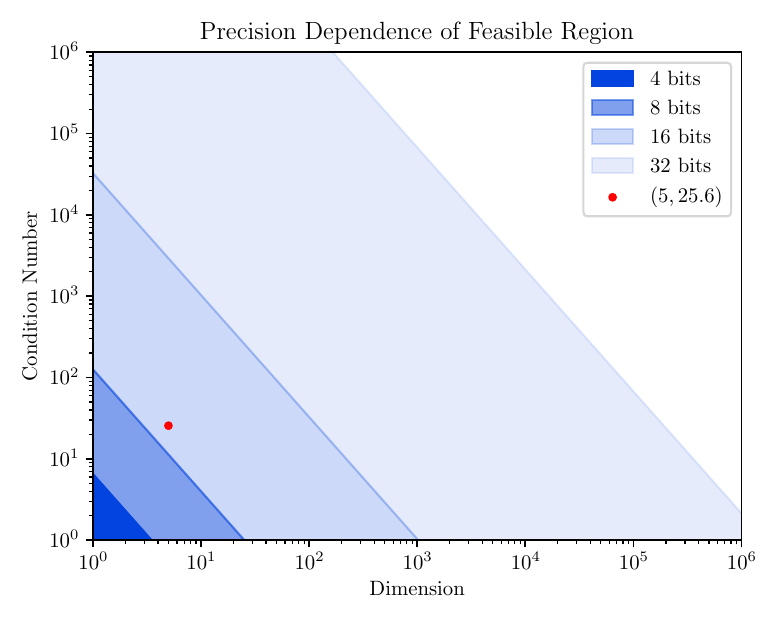}
    \caption{The feasible region for Thermies error mitigation for bit depth of $8$, $16$, and $32$. At fixed precision, there is a tradeoff between dimension and condition number of the target covariance matrix. The point $(5, 25.6)$ is plotted for reference, which corresponds to the matrix given in Eq. \eqref{finance-example-covariance-matrix}.}
    \label{fig:FeasibleRegion}
\end{figure}

Suppose the target covariance matrix $\Sigma^t$ has condition number at most $\kappa$, meaning \cite{trefethen1997numerical}
\begin{equation}
    \left| \lambda_\text{max}\right| \leq \kappa \left|\lambda_\text{min} \right|,
\end{equation}
where $\lambda_\text{max}$ is the eigenvalue having the largest absolute value and $\lambda_\text{min}$ is the eigenvalue having the smallest absolute value. Note that in dividing by $|\lambda_\text{min}|$, we are implicitly assuming $\Sigma^t$ is non-singular. The target covariance matrix may be scaled freely, so we can assume without loss of generality that 
\begin{equation}
   | \lambda_\text{min} | = d \varepsilon,
\end{equation}
which means $|\lambda_\text{max}| \leq \kappa d \varepsilon$. Suppose $\Sigma^b$ belongs to the set of nearest neighbors of $\Sigma^t$. Then any element $\Sigma^b_{ij}$ differs from $\Sigma^t_{ij}$ by at most $\varepsilon$. Let $\delta = \Sigma^b - \Sigma^t$. We then have
\begin{align}
    x^\intercal\Sigma^bx &= x^\intercal(\Sigma^t+ \delta)x  \\
    &=  x^\intercal\Sigma^t x+ x^\intercal\delta x \\
    & \geq x^\intercal \Sigma x -  \| \delta \|_\infty \\
    & \geq x^\intercal \Sigma x - d\varepsilon \geq 0.
\end{align}
This shows that all nearest neighbors $\Sigma^b$ of $\Sigma^t$ are positive semi-definite. 
 Recall the equivalence of matrix norms,
\begin{equation}
  \| \Sigma^t\|_\infty \leq \sqrt{d} \| \Sigma^t\|_2 = \sqrt{d} \lambda_\text{max},  
\end{equation}
where we have used the fact that $\Sigma^t\in \text{PSD}_d(\mathbb{R})$. This gives an upper bound on the magnitudes of all elements of the matrix,
\begin{equation}
    \left| \Sigma_{ij}^t\right| \leq \sqrt{d} \lambda_\text{max}.
\end{equation}
So, by the above reasoning,
\begin{equation}
\left|\Sigma_{ij}^t\right| \leq \kappa d^{3/2} \varepsilon.
\end{equation}
It follows that the nearest neighbors of $\Sigma^t$ have elements whose absolute values are at most $\kappa d^{3/2} \varepsilon + \varepsilon$. It follows that the bit depth $\xi$ of the (signed) covariances must be such that
\begin{equation}
    2^{\xi-1} - 1 \geq \kappa d^{3/2} +1,
 \end{equation}
so
\begin{equation}
\label{bit-depth-constraint}
\xi =\lceil \log_2(2\kappa d ^{3/2}+ 4)\rceil.
\end{equation}
If we take $\varepsilon = 1$, we can ensure that the neighbors are positive semi-definite by multiplying $\Sigma^t$ by a constant such that its largest element is $\kappa d^{3/2}$. An example of a covariance matrix describing the price fluctuations in a collection of financial instruments is taken from Ref.~\cite{chen_empirical_2015}, which reads
\begin{equation}
\label{finance-example-covariance-matrix}
\Sigma = 
\begin{pmatrix}
 1 & 0.652132 & 0.785365 & 0.608046 & 0.77665 \\
 0.652132 & 1 & 0.768745 & 0.544411 & 0.793887 \\
 0.785365 & 0.768745 & 1 & 0.640743 & 0.847228 \\
 0.608046 & 0.544411 & 0.640743 & 1 & 0.616502 \\
 0.77665 & 0.793887 & 0.847227 & 0.616502 & 1 \\
\end{pmatrix}.
\end{equation}
The matrix in Eq. \eqref{finance-example-covariance-matrix} has dimension $d=5$ and condition number $\kappa = 25.6$. The constraint of Eq. \eqref{bit-depth-constraint} is shown graphically in Fig. \ref{fig:FeasibleRegion}, and the point $(5, 25.6)$ is plotted for reference.

\section{Proof of Proposition~\ref{epsilon-independence-prop}}
\label{app:epsilon-independence}

Recall that the the covariance matching constraint is
\begin{equation}
 \Sigma^t = \sum_b w_b \Sigma^b.
\end{equation}
We define the perturbation matrices
\begin{equation}
    P^b = \frac{1}{\varepsilon}(\Sigma^b - \Sigma^t),
\end{equation}
and consequently $\Sigma^b = \Sigma^t + \varepsilon P^b$. In this form, the covariance-matching constraint is written
\begin{equation}
    \label{covariance-matching-perturbation}
    \sum_b w_b P^b = 0
\end{equation}

\begin{equation}
    f_a(x) = \sum_b w_b N_b \exp\left(- \frac{1}{2} x^\intercal(\Sigma^t+\varepsilon P^b)^{-1}x\right),
\end{equation}
where the normalization factor $N_b$ is
\begin{equation}
    N_b = \frac{1}{(2 \pi)^{d/2}\sqrt{|\Sigma^t + \varepsilon P^b|}}.
\end{equation}
In evaluating the derivative $\partial_\varepsilon f_a(x)$, there are terms which arise due to the normalization factors $N_b$ as well as terms which arise due to the exponential,
\begin{equation}
    \partial_\varepsilon f_a(x) = \sum_b w_b (\partial_\varepsilon N_b)\exp(\dots) + \sum_b w_b  N_b\partial_\varepsilon\exp(\dots).
\end{equation}
When we evaluate the derivative at $\varepsilon=0$, the factor which is not differentiated in each sum does not depend on $b$ so it can be pulled out of the sum, giving
\begin{equation}
\label{first-order-dependence-split}
    \left.\partial_\varepsilon f_a(x)\right|_{\varepsilon=0} = \exp(\dots)\sum_b w_b\left.\partial_\varepsilon N_b\right|_{\varepsilon=0}  + N\sum_b w_b\left.\partial_\varepsilon \exp(\dots)\right|_{\varepsilon=0} .
\end{equation}
We first evaluate the normalization terms,
\begin{align}
    \partial_\varepsilon N_b(\varepsilon) &= \partial_\varepsilon \frac{1}{(2\pi)^{d/2}|\Sigma^t + \varepsilon P^b|} \\
    & = -\frac{1}{2(2\pi)^{d/2}|\Sigma^t + \varepsilon P^b|^{3/2}}\partial_\varepsilon |\Sigma^t + \varepsilon P^b|\\
    & =  -\frac{1}{2(2\pi)^{d/2}|\Sigma^t + \varepsilon P^b|^{3/2}}|\Sigma^t + \varepsilon P^b|\:  \text{tr}\{(\Sigma^t + \varepsilon P^b)^{-1})P^b\},
\end{align}
where we have used Jacobi's formula. Evaluating at zero, we have
\begin{equation}
    \left. \partial_\varepsilon N_b(\varepsilon)\right|_{\varepsilon=0} = -\frac{1}{2(2\pi)^{d/2}|\Sigma^t|^{1/2}} \text{tr}\{\left(\Sigma^t\right)^{-1}P^b\}
\end{equation}
Therefore the first sum in Eq. \eqref{first-order-dependence-split} is
\begin{align}
    \sum_b w_b \left.\partial_\varepsilon N_b\right|_{\varepsilon=0} &=  - \frac{1}{2 (2\pi)^{d/2}|\Sigma^t|^{1/2}}\sum_b w_b \text{tr}\{(\Sigma^t)^{-1}P^b\} \\
    &= - \frac{1}{2 (2\pi)^{d/2}|\Sigma^t|^{1/2}} \text{tr}\left\{(\Sigma^t)^{-1}\sum_b w_b P^b\right\} \\
    & = 0,
\end{align}
where the last line follows from the covariance matching constraint \eqref{covariance-matching-perturbation}. Next we consider the derivative of the exponential factor. First note that the covariance matrix $\Sigma^t$ can be expanded in its eigenbasis as
\begin{equation}
    \Sigma^t = \sum_j \lambda_j v_jv_j^\intercal.
\end{equation}
We then have, to first order in $\varepsilon$,
\begin{align}
\exp\left(- \frac{1}{2} x^\intercal(\Sigma^t+\varepsilon P^b)^{-1}x\right) &= \exp\left(- \frac{1}{2} x^\intercal\left(\sum_j \lambda_j v_j v_j^\intercal+\varepsilon P^b\right)^{-1}x\right) \\
&=\exp\left(- \frac{1}{2} x^\intercal\left(\sum_j (\lambda_j +\varepsilon v_j^\intercal P^b v_j)v_j v_j^\intercal\right)^{-1}x\right) \\
    &=\exp\left(- \frac{1}{2} x^\intercal\left(\sum_j \frac{1}{\lambda_j +\varepsilon v_j^\intercal P^b v_j}v_j v_j^\intercal\right)x\right) \\
    &=\exp\left(- \frac{1}{2} x^\intercal\left(\sum_j \left[\frac{1}{\lambda_j} -\frac{\varepsilon v_j^\intercal P^b v_j}{(\lambda_j + \varepsilon v_j^\intercal P^b v_j)^2}\right]v_j v_j^\intercal\right)x\right) 
\end{align}
Therefore the derivative with respect to $\varepsilon$ is
\begin{align}
\partial_\varepsilon \exp\left(- \frac{1}{2} x^\intercal(\Sigma^t+\varepsilon P^b)^{-1}x\right) &=\exp(\dots)\left(\frac{1}{2} x^\intercal\left(\sum_j \frac{ v_j^\intercal P^b v_j}{(\lambda_j + \varepsilon v_j^\intercal P^b v_j)^2}v_j v_j^\intercal\right)x\right).
\end{align}
We can now compute the second sum in Eq. \eqref{first-order-dependence-split}, which is

\begin{align}
    \sum_b w_b \left.\partial_\varepsilon \exp (\dots)\right|_{\varepsilon=0} &= \sum_b w_b \exp(\dots) \left(\frac{1}{2} x^\intercal\left(\sum_j \frac{ v_j^\intercal P^b v_j}{\lambda_j^2}v_j v_j^\intercal\right)x\right) \\
    &=  \exp(\dots)  \left(\frac{1}{2} x^\intercal\left(\sum_j \frac{ v_j^\intercal \left(\sum_b w_b P^b \right)v_j}{\lambda_j^2}v_j v_j^\intercal\right)x\right) \\
    &= 0,
\end{align}
where the last line again follows from the covariance matching constraint \eqref{covariance-matching-perturbation}. This completes the proof that the covariance matching constraint is sufficient for the elimination of the first order dependence of $f_a(x)$ on $\varepsilon$ for $\varepsilon \to 0$.

\section{Thermies Protocol with Repetition}

The Thermies protocol described in the main text has the advantage of eliminating the first-order dependence of the approximate distribution  on the imprecision parameter, and moreover allows for perfectly matching the covariance matrix of the target distribution. However, the method has the disadvantage that one must draw a covariance matrix from the ensemble of nearest neighbors for each new sample which is required. Naturally, the question arises of whether high-quality samples can be obtained with fewer draws from the nearest neighbor ensemble. Here we present a version of the Thermies protocol which requires fewer samples from the nearest neighbor distribution, and analyze its performance. We call it the Thermies Protocol with Repetition. 

\medskip
\begin{tcolorbox}[title={Thermies Protocol with Repetition}]
\begin{enumerate}
    \item Compute the residual matrix $R = \Sigma^t/\varepsilon -   \floor{\Sigma^t/\varepsilon}$.
    \item For each pair $(i, j) \in \{1,2\dots d\}^2$ with $i \leq j$, draw a realization of the Bernoulli random variable $B_{ij}$ which has probabilities $\text{Pr}(B_{ij}=0) = 1-R_{ij}$ and $\text{Pr}(B_{ij}=1) = R_{ij}$. Also define $B_{ji} = B_{ij}$, resulting in a matrix of realizations $B \in \{0,1\}^{d \times d}$.
    \item Construct the matrix $\Sigma^B = \varepsilon (\floor{\Sigma^t/\varepsilon} +  B)$, and draw $n$ samples from the normal distribution $\mathcal{N}[0, \Sigma^B]$.
    \item Record the resulting samples as a realizations of the random vector $X$, without storing the results of the Bernoulli trials.
\end{enumerate}
\end{tcolorbox}
%\begin{figure}[b]
%    \centering
%    \includegraphics[width=0.5\textwidth]{figures/fig5.pdf}
 %   \caption{Root mean square (RMS) error between the target covariance matrix $\Sigma^t$ and the average of $M$ covariance matrices drawn from the nearest neighbor ensemble, as a function of $M$. The RMS error shows a dependence of roughly $M^{-1/2}$. This behavior is independent of dimension $d$.}
 %   \label{fig:ConvergenceAvgCov}
%\end{figure}

\medskip
The only difference between the Thermies method given before and this new one is that now $n$ samples are drawn from the normal distribution $\mathcal{N}[0,\Sigma^B]$ instead of only $1$. If the Thermies protocol with repetition is carried out $M$ times consecutively, then we end up with a sequence of $N = Mn$ realizations of $X$. Given that the nearest neighbor ensemble may be very large (in fact comprising $2^{(d^2 + d)/2}$ elements), one might be led to wonder whether a relatively very small selection of $M$ elements can provide a representative sample of this ensemble. In fact it can, which simply follows from Hoeffding's inequality, as we show in the next section. In particular, let $\bar{\Sigma}$ be the sample mean of a sequence of covariance matrices $\Sigma^1 \dots \Sigma^M$ drawn from the nearest-neighbor ensemble of a target covariance matrix $\Sigma^t$.
\begin{equation}
\label{sigma-bar-def_app}
\bar{\Sigma} = \frac{1}{M}\sum_{m=1}^M \Sigma^m.
\end{equation}
Choosing some $\delta >0$, Hoeffding's inequality says \cite{hoeffding1948class, hoeffding1963probability}
\begin{equation}
\label{hoeffding-bound_app}
    \text{Pr}\left( |\bar{\Sigma}_{ij} - \Sigma^t_{ij}| \geq \delta \right) \leq 2 \exp\left(-2 \frac{M \delta^2}{ \varepsilon^2} \right).
\end{equation}

If we hold the left-hand side of Eq. \eqref{hoeffding-bound} constant, and vary $M$ (while also holding $\varepsilon$ constant, we see that $\delta \propto M^{-1/2}$, which implies that the elementwise difference between $\bar{\Sigma}$ and $\Sigma^t$ goes with $M^{-1/2}$. Importantly, as there are no factors in Eq.~\eqref{hoeffding-bound} which involve $d$, this behavior is independent of dimension. In fact this is exactly what we see when random realizations of the nearest neighbor ensemble are realized and $\bar{\Sigma}$ is evaluated, as Fig.~\ref{fig:ConvergenceAvgCov} shows. We also see that a relatively small number of samples from the nearest neighbor ensemble is needed to get a good estimate of $\Sigma^t$.
\begin{figure}[t]
    \centering
    \includegraphics[width=0.5\textwidth]{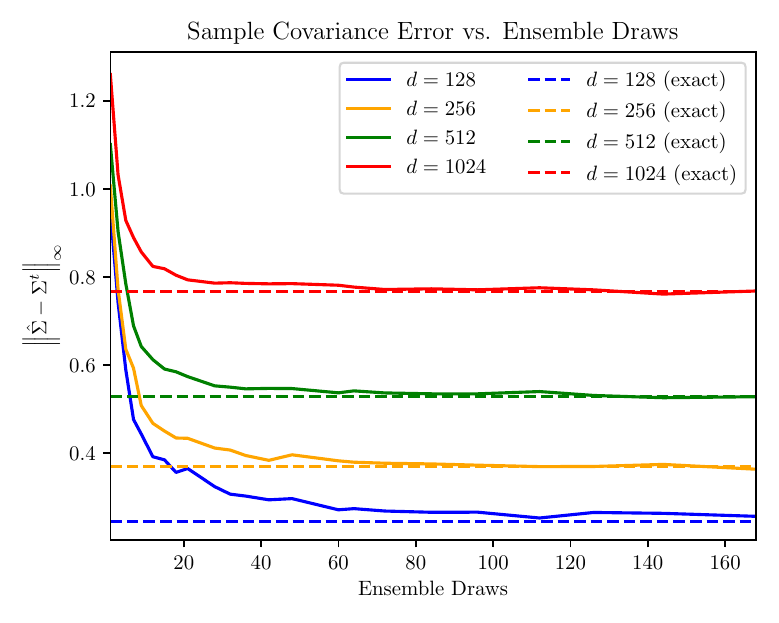}
    \caption{Infinity norm error between the sample covariance matrix $\hat{\Sigma}$ and the target covariance matrix $\Sigma^t$ for a sequence of $N = 2160$ samples, as a function of $M$, the number of covariance matrices drawn from the nearest neighbor ensemble. $\Sigma^t$ is a random positive diagonal matrix. Results averaged over $10$ random covariance matrices. Dashed lines show the infinity norm error between the target covariance and a sample covariance of a sequence drawn from the target distribution (equivalent to setting $M = N$).}
    \label{fig:ConvergenceDiagonal}
\end{figure}

This established, we consider the sample covariance of a sequence of samples obtained using the Thermies protocol with repetition. The sample covariance matrix is defined as
\begin{equation}
\label{sample-covariance-def}
    \hat{\Sigma}_{ij} = \frac{1}{N-1}\sum_{k=1}^N X^k_i X^k_j.s
\end{equation}
We would like the sample covariance matrix to converge towards the target covariance matrix $\Sigma^t$ as $N$ increases. We show in the next section that this occurs, the relevant inequality being
\begin{equation}
\label{combined-chebyshev-hoeffding-bound}
\text{Pr}\left(|\hat{\Sigma}_{ij} - \Sigma^t_{ij}|  \leq \delta\right) \geq \left(1- \frac{4 \bar{S}_{ij}}{N \delta^2}\right) \left( 1 - 2 \exp\left(-\frac{M\delta^2}{2\varepsilon^2} \right)\right),
\end{equation}

where $\delta >0$ and 
\begin{equation}
\label{s-bar-def}
    \bar{S}_{ij} = \frac{1}{M}\sum_k^{M} (\Sigma^k_{ij})^2 + \Sigma^k_{ii}\Sigma^k_{jj}.
\end{equation}

\section{Convergence of Thermies method via Hoeffding's bound}
\label{app:stats}

First suppose we draw a sequence of samples $X^{(1)} \dots X^{(n)}$ from a normal distribution $\mathcal{N}[0, \Sigma]$. We may then compute a sample covariance for this sequence, defined as
\begin{equation}
    \hat{\Sigma}_{ij}= \frac{1}{n-1} \sum_{k=1}^n X^k_i X^k_j \approx  \frac{1}{n} \sum_{k=1}^n X^k_i X^k_j ,
\end{equation}
where the approximate expression on the right-hand side is valid in the limit of large $n$. If this experiment is repeated many times, different outcomes will be obtained for $\hat{\Sigma}$, and we may ask what is the variance of this estimator. Because $\hat{\Sigma}_{ij}$ is the mean of a sequence of independent identically-distributed random variables, using the central limit theorem we can approximate it as a normally distributed random variable whose variance is the sum of the variances of the summands. Let $S_{ij}$ be the variance of $\braket{X_i X_j}$,
\begin{equation}
    S_{ij} = \braket{(X_i)^2 (X_j)^2} - \braket{X_i X_j}^2.
\end{equation}
The central limit theorem then gives
\begin{equation}
\label{sample-covariance-iid}
    \hat{\Sigma}_{ij} \sim \mathcal{N}\left[\Sigma_{ij},\frac{S_{ij}}{n} \right].
\end{equation}
Note that because $X^k$ is normally distributed, $S_{ij}$ can be expressed as
\begin{equation}
\label{S-from-Sigma}
S_{ij} = (\Sigma_{ij})^2 + \Sigma_{ii} \Sigma_{jj}.
\end{equation}

Now, imagine that we carry out the above procedure $M$ times, and each time we draw a different covariance matrix $\Sigma^m$ from the nearest neighbor ensemble of some target covariance matrix $\Sigma^t$. The total number of samples is then $N = n M$, and the sample covariance for the $N$ samples is
\begin{equation}
    \hat{\Sigma}_{ij} = \frac{1}{N-1} \sum_{m=1}^M \sum_{k=1}^{n}X^{(m)k}_i X^{(m)k}_j \approx \frac{1}{N} \sum_{m=1}^M \sum_{k=1}^{n}X^{(m)k}_i X^{(m)k}_j .
\end{equation}
We can then write the sample covariance matrix as the average of $M$ separate contributions $\hat{\Sigma}^{1}\dots \hat{\Sigma}^{M}$, one for each matrix drawn from the nearest neighbor ensemble
\begin{equation}
    \hat{\Sigma}_{ij} \approx \frac{1}{N} \sum_{m=1}^M n \hat{\Sigma}^{m}_{ij} = \frac{1}{M}\sum_{m=1}^M \hat{\Sigma}^{m}_{ij}.
\end{equation}
Therefore by the above analysis, the sample covariance matrix $\hat{\Sigma}$ is the mean of $M$ normally-distributed random variables $\hat{\Sigma}^m_{ij} \sim \mathcal{N}\left[\Sigma^m, S^m_{ij}/n\right]$. By the central limit theorem we then have
\begin{equation}
    \hat{\Sigma}_{ij} \sim \mathcal{N}\left[\bar{\Sigma}, \frac{\bar{S}_{ij}}{N}\right],
\end{equation}
where $\bar{\Sigma} = \frac{1}{M}\sum_{m=1}^M \Sigma^m$ and $\bar{S} = \frac{1}{M}\sum_{m=1}^M S^m$. According to Chebyshev's inequality \cite{serfling_probability_1974, saw_chebyshev_1984},
\begin{equation}
\label{Chebyshev-bound-app}
    \text{Pr}\left(|\hat{\Sigma}_{ij} - \bar{\Sigma}_{ij}| \geq k\sqrt{\frac{\bar{S}_{ij}}{N}}\right)\leq \frac{1}{k^2},
\end{equation}
which implies that the distance between $\hat{\Sigma}$ and $\bar{\Sigma}$ goes as $N^{-1/2}$ elementwise. The proximity of $\bar{\Sigma}$ to the target distribution $\Sigma^t$ can be bounded using Hoeffding's inequality. Note that the random variable $\Sigma^{m}_{ij}$ is bounded to an interval of length $\varepsilon$. Therefore Hoeffding's inequality says
\begin{equation}
    \text{Pr}\left( |\bar{\Sigma}_{ij} - \Sigma^t_{ij}| \geq \delta \right) \leq 2 \exp\left(-2 \frac{M \delta^2}{ \varepsilon^2} \right).
\end{equation}
Setting $\delta = k \sqrt{\frac{\bar{S}_{ij}}{N}}$ in Hoeffding's inequality gives
\begin{equation}
\label{hoeffding-bound-app}
    \text{Pr}\left( |\bar{\Sigma}_{ij} - \Sigma^t_{ij}| \geq k \sqrt{\frac{\bar{S}_{ij}}{N}} \right) \leq 2 \exp\left(-2 \frac{k^2\bar{S}_{ij}}{n \varepsilon^2} \right).
\end{equation}
Equation \eqref{hoeffding-bound-app} captures the intuitive notion that if each member of the ensemble is reused a large number of times (i.e. $n$ is large), the convergence of the sample covariance matrix is slower. Putting together Eqs. \eqref{Chebyshev-bound-app} and \eqref{hoeffding-bound-app}, we find
\begin{equation}
\text{Pr}\left(|\hat{\Sigma}_{ij} - \bar{\Sigma}_{ij}| \leq k\sqrt{\frac{\bar{S}_{ij}}{N}}\text{ and } |\bar{\Sigma}_{ij} - \Sigma^t_{ij}| \leq k\sqrt{\frac{\bar{S}_{ij}}{N}}\right) \geq \left(1 - \frac{1}{k^2} \right) \left(1 - 2\exp\left(-2 \frac{k^2 \bar{S}_{ij}}{n \varepsilon^2} \right) \right).
\end{equation}
If we then set $k\sqrt{\bar{S}_{ij}/N} = \delta/2$, we get
\begin{equation}
\text{Pr}\left(|\hat{\Sigma}_{ij} - \bar{\Sigma}_{ij}| \leq\delta/2 \text{ and } |\bar{\Sigma}_{ij} - \Sigma^t_{ij}| \leq \delta/2\right) \geq \left(1- \frac{4 \bar{S}_{ij}}{N \delta^2}\right) \left( 1 - 2 \exp\left(-\frac{M\delta^2}{2\varepsilon^2} \right)\right).
\end{equation}
Using the triangle inequality we can write
\begin{equation}
\label{combined-chebyshev-hoeffding-bound-app}
\text{Pr}\left(|\hat{\Sigma}_{ij} - \Sigma^t_{ij}|  \leq \delta\right) \geq \left(1- \frac{4 \bar{S}_{ij}}{N \delta^2}\right) \left( 1 - 2 \exp\left(-\frac{M\delta^2}{2\varepsilon^2} \right)\right).
\end{equation}
The above analysis implies that the error $\delta$ roughly goes with $M^{-1/2}$ when $N$ is held constant, which can be seen by holding the left-hand side of \eqref{combined-chebyshev-hoeffding-bound-app} constant and varying $\delta$.

\section{Implementation on thermodynamic hardware}
\label{app:hw}

Here we elaborate on our implementation of the Thermies protocol on the thermodynamic computing hardware described in Ref.~\cite{melanson2023thermodynamic}. This hardware is composed of 8 unit cells (i.e., 8 RLC circuits) that are all-to-all coupled to each other. It was used in Ref.~\cite{melanson2023thermodynamic} to invert 8x8 matrices, and that is the primitive that we focus on applying Thermies to here in this work.

This hardware encodes the covariance matrix $\Sigma$ (or precision matrix $\Sigma^{-1}$) in the values of its capacitances that can only take on specific values. Specifically, the device can only sample from zero-mean Gaussian distributions of covariance matrices with diagonal elements in $\{ 1.0, 3.2,  4.3, 6.5\}$ and off-diagonal elements in $\{ -0.47, 0.0, 0.47 \}$.

Due to the relatively low precision and specific construction of this device, the Thermies protocol needs to be modified slightly to accommodate. The difference between the protocol as described in the main text and the one needed for the experimental realization is that the imprecision, $\varepsilon$, is not uniform and has limited range.

The variant of the Thermies protocol used on this device is as follows:
\begin{itemize}
    \item Scale the precise target matrix such that all the elements are within the hardware limits.
    \item Determine the pair of nearest neighbors for each element. That is, the closest hardware values that are above and below each precise element.
    \item Calculate the weights for each element as
        \begin{align*}
            w_{ij} = \frac{|\sigma^\mathrm{t}_{ij} - \sigma^\mathrm{b-}_{ij}|}{\sigma^\mathrm{b+}_{ij} - \sigma^\mathrm{b-}_{ij}},
        \end{align*}
    where $\sigma^\mathrm{t}_{ij}$ are the target matrix elements, while $\sigma^\mathrm{b+}_{ij}$ ($\sigma^\mathrm{b-}_{ij}$) are the nearest higher (lower) imprecise matrix elements.
    \item Draw a member of the nearest neighbor matrices by transforming the precise target matrix into the imprecise matrix by probabilistically rounding up each element with probability $w_{ij}$. Repeat this $M$ times.
    \item Gather $n$ samples from each member of the $M$ members of the ensemble.
    \item Combine all $nM$ samples and compute their covariance matrix.
\end{itemize}

For the specific implementation shown in Fig.~\ref{fig:PCB_Implementation}, we employ the thermodynamic matrix inversion algorithm from Ref.~\cite{aifer2023thermodynamic}, which involves estimating the covariance matrix for a collection of samples in order to invert the precision matrix. When applying Thermies to this task, we employ 4 ensemble draws (i.e., 4 draws from the set of nearest neighbors) and then average over those ensembles draws to obtain the results.

\end{appendix}

\end{document}